\documentclass[conference]{IEEEtran}
\IEEEoverridecommandlockouts
\usepackage{cite}
\usepackage{amsmath,amssymb,amsfonts}
\usepackage{dsfont}
\usepackage{algorithmic}
\usepackage{graphicx}
\usepackage{textcomp}
\usepackage{xcolor}

\usepackage{mathtools}

\DeclarePairedDelimiter{\floor}{\lfloor}{\rfloor}
\usepackage{amsthm}
\usepackage{arydshln}
\usepackage{multirow}
\usepackage{calc}
\usepackage{blkarray}
\usepackage{url}
\theoremstyle{definition}
\newtheorem{theorem}{Theorem}

\newtheorem{example}{Example}
\newtheorem{remark}{Remark}

\usepackage{graphicx,cite}
\usepackage{dblfloatfix}
\usepackage[justification=centering]{caption}
\usepackage{blindtext, graphicx, amsfonts,
	amssymb,multirow,epstopdf}
\usepackage[linesnumbered,ruled]{algorithm2e}
\def\BibTeX{{\rm B\kern-.05em{\sc i\kern-.025em b}\kern-.08em
    T\kern-.1667em\lower.7ex\hbox{E}\kern-.125emX}}

\newcommand{\calB}{\mathcal{B}}
\newcommand{\calC}{\mathcal{C}}
\newcommand{\calA}{\mathcal{A}}

\newcommand{\calU}{\mathcal{U}}
\newcommand{\calV}{\mathcal{V}}

\newcommand{\bfA}{\mathbf{A}}

\newcommand{\bfx}{\mathbf{x}}
\newcommand{\bfw}{\mathbf{w}}
\newcommand{\bfc}{\mathbf{c}}

\begin{document}

\title{$C^{3}LES$: Codes for Coded Computation that Leverage Stragglers}

\author{\IEEEauthorblockN{Anindya B. Das, Li Tang and Aditya Ramamoorthy}
\IEEEauthorblockA{Department of Electrical and Computer Engineering \\
Iowa State University\\
Ames, IA 50010, U.S.A. \\
\{abd149,litang,adityar\}@iastate.edu}
\thanks{This work was supported in part by the National Science Foundation (NSF) under grant CCF-1718470.}
}

\maketitle

\begin{abstract}
In distributed computing systems, it is well recognized that worker nodes that are slow (called stragglers) tend to dominate the overall job execution time. Coded computation utilizes concepts from erasure coding to mitigate the effect of stragglers by running ``coded'' copies of tasks comprising a job. Stragglers are typically treated as erasures in this process.

While this is useful, there are issues with applying, e.g., MDS codes in a straightforward manner. Specifically, several applications such as matrix-vector products deal with sparse matrices. MDS codes typically require dense linear combinations of submatrices of the original matrix which destroy their inherent sparsity. This is problematic as it results in significantly higher processing times for computing the submatrix-vector products in coded computation. Furthermore, it also ignores partial computations at stragglers.

In this work, we propose a fine-grained model that quantifies the level of non-trivial coding needed to obtain the benefits of coding in matrix-vector computation. Simultaneously, it allows us to leverage partial computations performed by the straggler nodes. For this model, we propose and evaluate several code designs and discuss their properties.
\end{abstract}

\begin{IEEEkeywords}
Distributed computing, stragglers
\end{IEEEkeywords}

\section{Introduction}
Distributed computation plays a major role in several problems in machine learning. For example, large scale matrix-vector multiplication is repeatedly used in gradient descent which is typically used in high dimensional machine learning problems. The size of the underlying matrices makes it impractical to perform the computation on a single computer (both from a speed and a storage perspective). Thus, the computation is typically subdivided into smaller tasks that are run in parallel across multiple worker nodes.



In these systems the overall execution time is typically dominated by the speed of the slowest worker. Thus, the presence of stragglers (as these slow workers are called) can negatively impact the performance of distributed computation. In recent years, techniques from coding theory \cite{lee2018speeding,dutta2016short,yu2017polynomial,tandon2017gradient} have been used to mitigate the effect of stragglers for problems such as matrix-vector and matrix-matrix multiplication. For instance, the work of \cite{lee2018speeding} proposes to partition the computation of $\bfA \bfx$ by first splitting $\bfA^T = [\bfA_1^T~\bfA_2^T]^T$ into an equal number of rows and assigning three workers, the task of computing $\textbf{A}_1\textbf{x}$, $\textbf{A}_2\textbf{x}$ and $\left(\textbf{A}_1+\textbf{A}_2\right)\textbf{x}$, respectively. Evidently, the load on each node is half of the original job. Furthermore, it is easy to see that $\bfA \bfx$ can be recovered as soon as any two workers complete their tasks (with some minimal post-processing). Thus, this system is resilient to one straggler. The work of \cite{yu2017polynomial}, poses the multiplication of two matrices in a form that is roughly equivalent to a Reed-Solomon code. In particular, each worker node's task (which is multiplying smaller submatrices) can be imagined as a coded symbol. As long as enough tasks are complete, the master node can recover the matrix product by polynomial interpolation.

For such systems we can define a so-called recovery threshold, which is defined as the minimum value of $\tau$, such that the master node can obtain the result as long as {\it any} $\tau$ workers complete their tasks. Thus, at the top level, in these systems stragglers are treated as the equivalent of erasures in coding theory, i.e., the assumption is that no useful information can be obtained from the stragglers.

While these are interesting ideas, there are certain issues that are ignored in the majority of prior work (see \cite{kiani2018exploitation, mallick2018rateless, wang2018coded} for some exceptions). Firstly, several practical cases of matrix-vector or matrix-matrix multiplication involve sparse matrices. Using MDS coding strategies in a straightforward manner will often destroy the sparsity of the matrices being processed by the worker nodes. In fact, as noted in \cite{wang2018coded}, this can cause the overall job execution time to actually go up rather than down. Secondly, in the distributed computation setting, we make the observation that it is possible to leverage partial computations that are performed by the stragglers. Thus, a slow worker may not necessarily be a useless worker.

\subsection{Main Contributions}
\begin{itemize}
\item In this work we present a more fine-grained model of the distributed matrix-vector multiplication that allows us to (i) leverage partial computations performed by stragglers and (ii) impose constraints on the extent to which coding is allowed in the solution. 
Our formulation leads to some new questions in the domain of code design that to our best knowledge have not been investigated systematically in the literature before.
\item We present two models in our work. In the first model, the tasks assigned to the workers are uncoded, whereas in the second model we allow for a user specified fraction of coded tasks. In both cases, we present bounds on the amount of computation that the workers need to perform in the worst case and the straggler resilience of the system. We also present matching construction schemes in some cases. We emphasize that the uncoded model applies in general to {\it any} computation problem, and the bounds and constructions hold in significant generality for that case.
\end{itemize}

\begin{figure}[t]
\centering
\includegraphics[width=55mm]{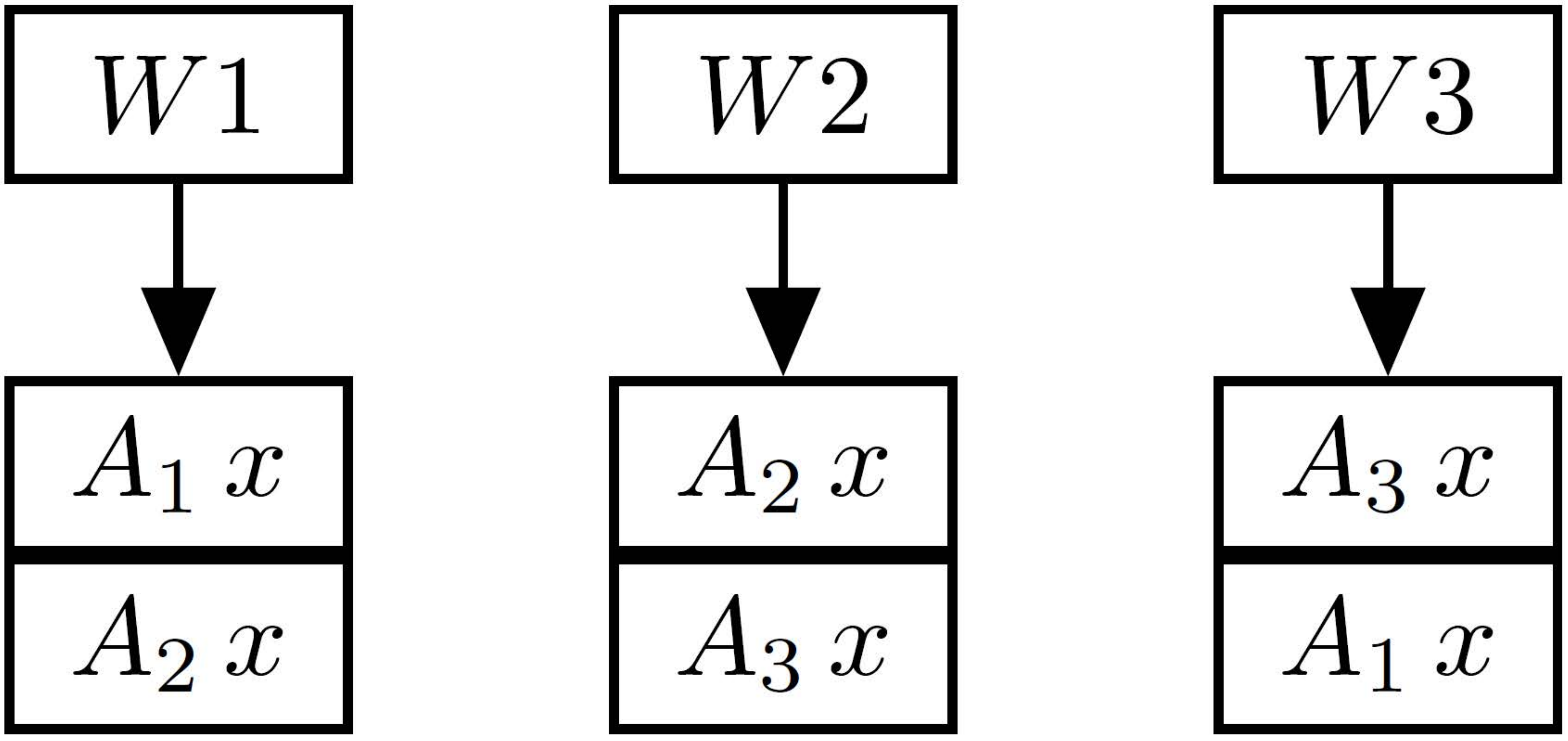}
\caption{\small Matrix $A$ is divided into three submatrices. Each worker is assigned two of the submatrices.}
\label{fig1}
\end{figure}

\section{Problem Formulation}
\label{sec:prob_form}

We consider a scenario where a master node has a matrix $\bfA$ and a vector $\bfx$ and needs to compute $\bfA \bfx$. The computation needs to be carried out in a distributed fashion over $n$ nodes. Each node receives a certain fraction (denoted by $\gamma$) of the rows of $\bfA$ and the vector $\bfx$. The node is responsible for computing the product of its assigned submatrix and $\bfx$.


We assume that the storage fraction $\gamma$ can be expressed as $\ell/\Delta$ where both $\ell$ and $\Delta$ are integers. In this work, we assume that $\bfA$ is large enough so that we can choose any large enough value of $\Delta$. Following this, we partition the rows of $\bfA$ into $\Delta$ submatrices denoted $\bfA_1, \dots, \bfA_\Delta$; we will also refer to these as the blocks of $\bfA$. Each node is assigned the equivalent of $\ell$ block rows. The assigned block rows can simply be subsets of $\{\bfA_1, \dots, \bfA_\Delta\}$; in this case we call the solution ``uncoded". Alternatively, the assigned block rows can be suitably chosen functions of $\{\bfA_1, \dots, \bfA_\Delta\}$; in this case we call the solution ``coded". Each worker node processes its assigned block rows sequentially from the top to the bottom. In particular, if a node is currently processing the $i$-the block row ($1 \leq i \leq l$), then it has already processed blocks $1$ through $i-1$. As we shall show, the processing order matters in this problem.

We assume that each time a node computes the block product (with $\bfx$) it transmits the result to the master node. We enforce the requirement that the master node should be able to recover $\bfA \bfx$ as long it receives {\it any} $Q$ block products from the worker nodes. This formulation subsumes treating stragglers as non-working nodes. Indeed, suppose that we want a system that is resilient to $s$ stragglers. Then, a sufficient condition would be that $Q \leq (n-s) \ell$ in our system.


\begin{example}
\label{eg:initial_example}
Consider a system with $n=3$ worker nodes with $\gamma = 2/3$. We partition $\bfA$ into $\Delta=3$ row blocks and the assignment of blocks to each node is shown in Fig. \ref{fig1} (this is an uncoded solution). We emphasize that the order of the computation also matters here, i.e., worker node $1$ (for example) computes $\bfA_1 \bfx$ first and then $\bfA_2 \bfx$. For the specific assignment it is clear that the computation is successful as long as any four block products are returned. Thus, for this system $Q = 4$.

On the other hand, Fig. \ref{fig2} demonstrates a coded solution, where the assignment in the second block rows of the workers are some suitably chosen functions of the elements of $\{\bfA_1 \bfx, \bfA_2 \bfx, \bfA_3 \bfx\}$. For this assignment, it is obvious that the master can recover $\bfA \bfx$ as long as any three block products are returned by the workers, so in this system $Q = 3$.

\begin{figure}[t]
\centering
\includegraphics[width=75mm]{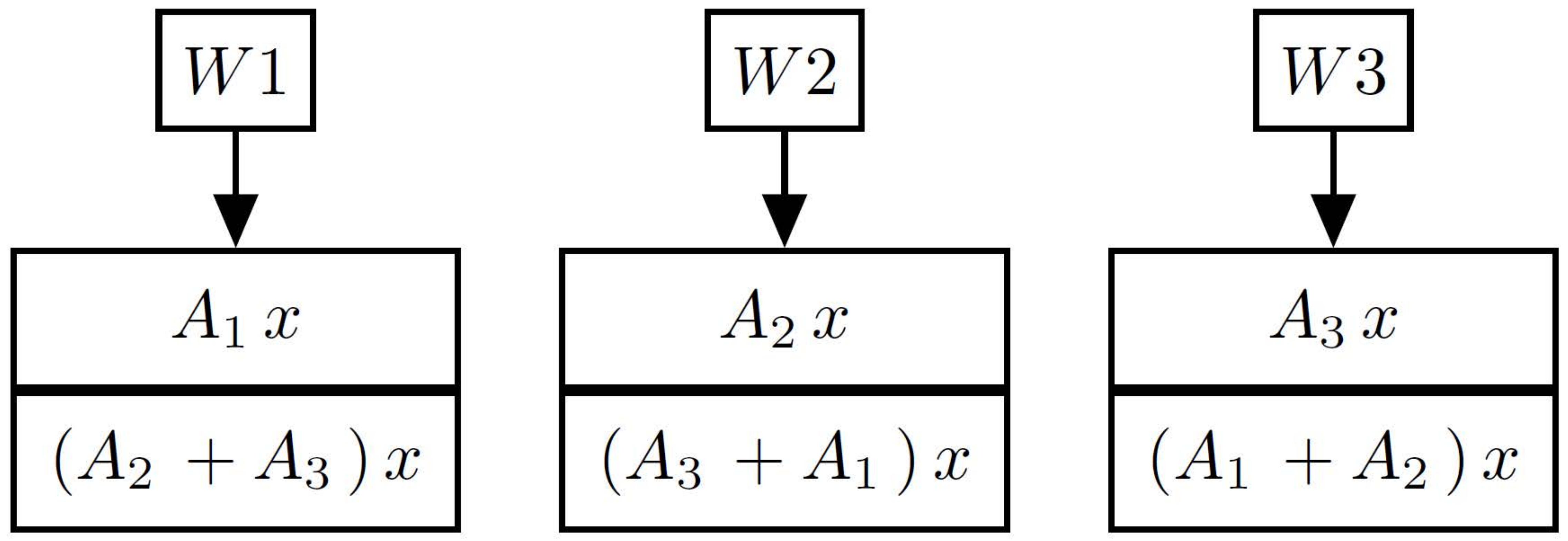}
\caption{\small Matrix $A$ is divided into three submatrices. Each worker is assigned two submatrices one of which is coded. }
\label{fig2}
\end{figure}

\end{example}
For any time $t$, we let $w_i(t)$ represent the state of computation of the $i$-th worker node, i.e., $w_i(t)$ is a positive integer between $0$ and $\ell$ which represents the number of block rows that have been processed by worker node $i$. Thus, our system requirement states as long as $\sum_{i=1}^n w_i(t) \geq Q$, the master node should be able to determine $\bfA \bfx$. As $\Delta$ is a parameter that can be chosen, our objective is to minimize the value of $Q/\Delta$ for such a system. This formulation minimizes the overall computation performed by the worker nodes.


\begin{remark}
It is important to note that the uncoded formulation applies to {\it any} computation job that can be subdivided into $\Delta$ tasks. In particular, the structure of matrix-vector multiplication is not important in this context. Thus, the discussion in the subsequent sections for the uncoded setup applies in significantly more generality.
\end{remark}

\section{Uncoded Scheme}
\label{sec:uncoded_scheme}
The advantage of an uncoded scheme is that it does not require any computation from the master to recover $\bfA \bfx$ because the assignment to any worker is simply a subset of $\{\bfA_1 \bfx, \bfA_2 \bfx, \dots, \bfA_\Delta \bfx \}$. 


To avoid trivialities, we emphasize that each worker node only contains at most one copy of each block row. In what follows, we use $r$ to represent the replication factor of each $\bfA_i$. Thus, each $\bfA_i$ appears $r$ times across all worker nodes. We use the notation $\langle n,\ell,\Delta,r \rangle$-uncoded system to represent an uncoded system with the corresponding parameters.
\begin{theorem}
\label{thm:straggler_resilience}
Consider an $\langle n,\ell,\Delta,r \rangle$-uncoded system. If the system needs to be resilient to $s$ stragglers, then $r \geq s+1$ and $n \gamma = r$.
\end{theorem}
\begin{proof}
It is evident that $r \geq s+1$ since we need at least one copy of each block row to be present even in the presence of $s$ stragglers. Next, the total number of symbols across all nodes equals $\Delta r$. A simple double counting argument, yields the equality $n \ell = \Delta r$ which further implies that $n \gamma = r$ by the definition of $\gamma$.
\end{proof}

\begin{figure}[t]
\centering
\includegraphics[width=85mm]{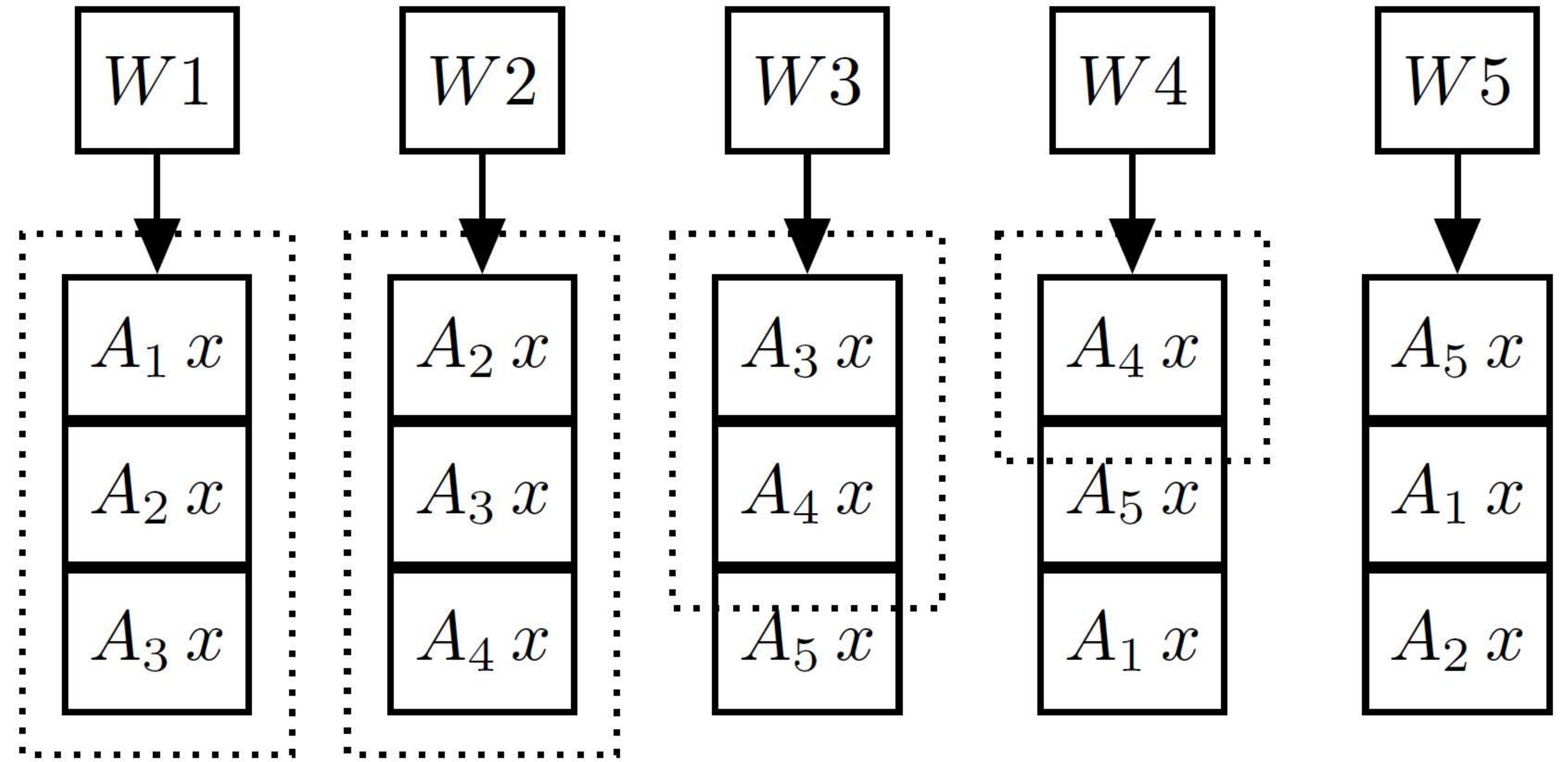}
\caption{\small A $\langle n, \ell, \Delta,r \rangle = \langle 5, 3, 5, 3\rangle$-uncoded system designed using Algorithm \ref{Alg:Cyclic_Uncoded}. The dotted blocks show the blocks that can be processed in the worst case without processing $A_5$.}
\label{fig3}
\end{figure}

\begin{theorem}
\label{thm:Q_by_Delta_bound}
Consider an $\langle n,\ell,\Delta,r \rangle$-uncoded system. Then, $Q \geq \max(\Delta, \Delta r - \frac{r}{2}\, (l+1) + 1)$.
\end{theorem}
\begin{proof}
For the system under consideration the master node requires each block product, $\bfA_j \bfx$ where $j = 1 , 2 , ... , \Delta$ to be computed at least once by the $n$ worker nodes. It is evident that the system needs to process at least $\Delta$ blocks, so that $Q \geq \Delta$.
Let $Q_j$ represent the maximum number of block rows that are processed in the worst case without obtaining $\bfA_j \bfx$ (see Fig. \ref{fig3} for an example). It is evident in this case that $Q = \max_{j=1, \dots, \Delta} Q_j + 1$.

Our strategy is to calculate the average $\overline{Q} = \sum\limits_{j=1}^{\Delta} Q_j/\Delta$ and use the simple bound $Q \geq \overline{Q} + 1$. Toward this end, note that for any uncoded solution, we can calculate $\sum_{j=1}^\Delta Q_j$ in a different way. For a worker $i$, there are $\ell$ assigned block rows and $\Delta - \ell$ do not appear in it. Thus, in the calculation of $\sum_{j=1}^\Delta Q_j$, worker node $i$ contributes

$$ \sum\limits_{k=1}^{\ell} (k-1) + (\Delta - \ell) \ell,$$
which is clearly independent of $i$. Therefore,

\begin{align*}
\label{eq3}
\overline{Q} &= n\left[ \frac{\sum\limits_{k=1}^{\ell} (k-1) + (\Delta - \ell) \ell }{\Delta} \right]\\
&=  n \ell - \frac{n \ell}{2 \Delta} (\ell+1) \; = \;  \Delta r - \frac{r}{2} (\ell+1),
\end{align*}
where we used $n \ell = \Delta r$ in the last step above. 
\end{proof}


Thus, Theorem \ref{thm:straggler_resilience} gives an upper bound on the number of stragglers that a system tolerates and Theorem \ref{thm:Q_by_Delta_bound} provides a lower bound on $Q$. Both these results can be treated as benchmarks for an uncoded scheme. We now propose a construction (see Algorithm \ref{Alg:Cyclic_Uncoded}) which meets both these bounds. The basic idea in Algorithm \ref{Alg:Cyclic_Uncoded} is to set $\Delta = n$ and place the block rows in a cyclic fashion (see Fig. \ref{fig3}).

\begin{algorithm}[t]
	\caption{Cyclic Uncoded Scheme}
	\label{Alg:Cyclic_Uncoded}
   \SetKwInOut{Input}{Input}
   \SetKwInOut{Output}{Output}
   \Input{Matrix $\bfA$ and vector $\bfx$, $n$-number of worker nodes, replication factor $r$.}
   Set $\Delta = n$ and $\ell = r$. Partition $\bfA$ into $\Delta$ block rows $\bfA_1, \dots, \bfA_\Delta$\;
   \For{$i\gets 1$ \KwTo $n$}{
   Assign $\bfA_{i}, \bfA_{i+1}, \dots, \bfA_{i + \ell -1}$ from top to bottom (subscripts reduced modulo $\Delta$) to worker node $i$
   }
   \Output{$\langle n,\ell,\Delta,r \rangle$ uncoded system with stragger resilience $r-1$ and optimal $Q/\Delta$.}
\end{algorithm}

\begin{theorem}
\label{thm:cyclic_uncoded_scheme}
The Cyclic Uncoded Scheme ({\it cf.} Algorithm \ref{Alg:Cyclic_Uncoded}) is resilient to $(r-1)$ stragglers and meets the bound in Theorem \ref{thm:Q_by_Delta_bound}.
\end{theorem}

\begin{proof}

It is evident from the cyclic nature of the construction that each block row appears $r$ times in $r$ different worker nodes. Thus, the scheme is resilient to $r-1$ stragglers. Next, we show that $Q_j$ (defined in the proof of Theorem \ref{thm:Q_by_Delta_bound}) in this construction is the same for all block rows $\bfA_j$, $j =1, \dots, \Delta$. To see this we note that the maximum number of block rows that can be processed without processing $\bfA_j$ can be calculated as follows.

\begin{itemize}
\item We can process the $(i-1)$ rows in worker node $[(j-i) \mod n] + 1$ for $i = 1, \dots, r$. 
\item All the block rows in other $n-r$ workers can be processed as well.
\end{itemize}

Thus, we have
\begin{align*}
Q_j & = \left[ 1 + 2 + 3 + ... + (\ell - 1) \right] + ( n - r ) \ell \\
& = \frac{(\ell-1)\ell}{2} + n \ell - \ell^2  \text{~~~(as $r = \ell$)}\\
& = n \ell - \frac{\ell}{2}( \ell + 1 ).
\end{align*} Again, using the fact that $n = \Delta$ and $\ell = r$, we have
\begin{align*}
Q_j = \Delta r - \frac{r}{2}( \ell + 1 ).
\end{align*}
Thus, $Q_j$ is independent of $j$ and therefore $Q = \Delta r - \frac{r}{2}( \ell + 1 ) + 1$.
\end{proof}

\begin{example}
\label{eg:uncoded_cons}
As an example, consider Fig. \ref{fig3}, where we have  $\Delta = n = 5$ and $\ell = r = 3$. The scheme is resilient to $(r - 1) = 2$ stragglers and it can be verified that $\bfA \bfx$ can be computed once any $Q=10$ block rows have been processed. 
\end{example}
\section{Coded Scheme}
We now explore coded schemes in our setting. As demonstrated in Example \ref{eg:initial_example}, the value of $Q$ in the coded scenario can be strictly lesser than in the uncoded case. A simple solution to this problem is to use MDS (maximum distance separable) codes as was done in some of the inital papers \cite{lee2018speeding} in this area. Namely, one could use an $(n\ell,\Delta)$-MDS code for some value of $\Delta$ and $n \ell \geq \Delta$. The block rows $\bfA_1, \dots, \bfA_\Delta$ will be combined using the corresponding generator matrix (with an appropriate mapping from the finite field to the real field). It is clear that in this case, the master node can compute $\bfA \bfx$ as long as any $\Delta$ rows are processed. However, such codes may require rather dense linear combination of the rows of $\bfA$ and this may incur a significant overhead in the computation performed by the worker nodes, especially in the practical case when $\bfA$ is sparse to begin with.

Accordingly, in this part of the work we are interested in examining schemes where we can specify the fraction of coded blocks.
We now assume that each node receives a $\gamma = \gamma_u + \gamma_c$ fraction of the rows of $\bfA$, where $\gamma_u$ and $\gamma_c$ correspond to the uncoded and coded parts respectively.
The replication factor for the uncoded portion is $r_u$ and the number of uncoded block assignments in a worker is denoted as $\ell_u = \Delta \gamma_u$. This implies that $ n \ell_u = \Delta r_u$. We let $\ell_c = \Delta \gamma_c$ represent the number of coded blocks in each worker.


For our bounds we assume that each processed coded block is useful to the master node. This can be ensured by our choice of coding coefficients that are obtained from Cauchy matrices of appropriate dimension. In what follows, we consider two different schemes. In one case the coded blocks appear at the bottom of each node while in the other case the coded blocks appear at the top. The first case leans towards easier decoding by the master node, whereas the second one aims to minimize the computation performed by the worker nodes. We use the notation $\langle n, \ell_u, \ell_c, \Delta, r_u \rangle$-bottom and $\langle n, \ell_u, \ell_c, \Delta, r_u \rangle$-top to refer to the corresponding systems. The value of $Q$ for these systems will be denoted by $Q_{cb}$ and $Q_{ct}$ respectively (the subscripts are self-explanatory).

\subsection{Coded blocks at the bottom}
\label{sec:coded_blocks_bottom}
In this case, the results of Section \ref{sec:uncoded_scheme} (See Theorem \ref{thm:Q_by_Delta_bound}) immediately imply that $Q_{cb} \geq \max (\Delta, \Delta r_u - \frac{r_u}{2} (\ell + 1) +1)$ (the subscript denotes that the coded blocks appear at the bottom).  This follows by applying the previous arguments to the uncoded part of the solution. A construction that meets these bounds is outlined in Algorithm \ref{Alg:Cyclic_Partially_Uncoded}.
The algorithm uses a Cauchy matrix of dimension $n \ell_c \times \Delta$.


\begin{algorithm}[t]
	\caption{Cyclic Coded at Bottom Scheme}
	\label{Alg:Cyclic_Partially_Uncoded}
   \SetKwInOut{Input}{Input}
   \SetKwInOut{Output}{Output}
   \Input{Matrix $\bfA$ and vector $\bfx$, $n$-number of worker nodes, total storage capacity fraction $\gamma$, replication factor for uncoded portion $r_u$.}
   Set $\Delta = n$, $\ell_u = r_u$, $\ell = \gamma \Delta$, $\ell_c = \ell - \ell_u$. Determine a Cauchy matrix $\calC$ of dimension $n \ell_c \times \Delta$\;
   Partition $\bfA$ into $\Delta$ block rows $\bfA_1, \dots, \bfA_\Delta$\;
   \For{$i\gets 1$ \KwTo $n$}{
   Assign $\bfA_{i}, \bfA_{i+1}, \dots, \bfA_{i + \ell_u -1}$ from top to bottom (subscripts reduced modulo $\Delta$) to worker node $i$\;
    \For{$j \gets 1$  \KwTo $\ell_c$}{
        Define $T = \{i, i+1, \dots, i + \ell_u - 1\} \mod \Delta$\;
        Pick a row $\bfc$ of $\calC$  and assign coded block $\sum_{k=1}^\Delta \bfc_k \mathds{1}_{k \notin T} \bfA_k$\;
        Remove $\bfc$ from $\calC$, so $\mathcal{C} \gets \mathcal{C} \backslash \{\bfc\}$ \;
    }
   }
   \Output{$\langle n,\ell_u,\ell_c,\Delta,r_u \rangle$-bottom system.}
\end{algorithm}

\begin{theorem}
\label{thm:Q_partially_coded_end}
The scheme in Algorithm \ref{Alg:Cyclic_Partially_Uncoded} satisfies $Q_{cb} = \max (\Delta, \Delta r_u - \frac{r_u}{2} (\ell_u + 1) +1)$. Furthermore, it is resilient to  $\floor[\bigg]{\frac{n^2 \gamma_c + n \gamma_u - 1}{n \gamma_c + 1}}$ stragglers.
\end{theorem}
\begin{proof}

We need to show that for any pattern of $Q_{cb}$ blocks the master node can decode $\bfA \bfx$. Towards this end, from the discussion in Section \ref{sec:uncoded_scheme}, we know that any pattern of $Q_{cb}$ uncoded blocks allows the recovery of $\Delta$ distinct blocks. In other words for any computation state vector $ \bfw(t) =[w_1(t)~w_2(t)~ \dots~w_n(t)]$ such that $w_i(t) \leq \ell_u$ and $\sum_{i=1}^n w_i(t) \geq Q_{cb}$ the master node can decode. Now, consider a vector $\bfw'(t)$ such that (w.l.o.g.) $w'_1(t) \dots w'_\alpha(t) \geq \ell_u + 1$ and $w'_{\alpha+1}(t), \dots, w'_n(t) \leq \ell_u$ and $\sum_{i=1}^n w'_i(t) \geq Q_{cb}$, i.e., the first $\alpha$ worker nodes process coded blocks whereas the others do not. It is not too hard to determine a different vector $\tilde{\bfw}(t)$ with the following properties.
\begin{align*}
\tilde{w}_i(t) &= \begin{cases}
\ell_u & ~1 \leq i \leq \alpha,\\
w'_i(t) + \beta_i & \alpha+1 \leq i \leq n,
\end{cases}
\end{align*}
where $\beta_i$'s are positive integers such that $w'_i(t) + \beta_i \leq \ell_u$ and $\sum_{i=1}^n \tilde{w}_i(t) = Q_{cb}$. Thus, $\tilde{\bfw}(t)$ corresponds to a pattern of $Q_{cb}$ uncoded blocks that recovers $\Delta$ distinct blocks.

Now, we compare the vectors $\bfw'(t)$ and $\tilde{\bfw}(t)$. Let the uncoded symbols in $\bfw'(t)$ be denoted by the set $\calA$. Then the set of uncoded symbols in $\tilde{\bfw}(t)$ can be expressed as $\calA \cup \calB$ where the set $\calB$ results from the transformation above. It is evident that for computation state vector $\bfw'(t)$ the master node has $\sum\limits_{i=1}^{\alpha} \left( w'_i(t) - \ell_u \right)$ equations with $\Delta - |\calA|$ variables. Now,
\begin{align*}
\sum\limits_{i=1}^{\alpha} \left( w'_i(t) - \ell_u \right) \geq |\calB|
&\geq |\calB \setminus \calA| = \Delta - |\calA|.
\end{align*}
In particular, this establishes that we have at least as many equations as variables. As, any square submatrix of a Cauchy matrix is invertible, we have the required result.

To establish the straggler resilience of our construction we identify the set of worker nodes that contain the least number of uncoded blocks. With some work (the details appear in the full version of the paper) it can be established that any $k$ workers have at least $\min(\ell_u + k - 1,\Delta)$ uncoded blocks. For our construction, these correspond to picking $k$ consecutive worker nodes. Thus, we are trying to determine the minimum value of $k$ such that
%
\begin{align*}
\ell_u + (k - 1) + k (\ell - \ell_u) \geq \Delta
\end{align*} which further implies
\begin{align*}
k \geq \frac{n - \ell_u + 1}{\ell - \ell_u + 1} = \frac{n - n \gamma_u + 1}{n \gamma - n \gamma_u + 1}
\end{align*} as $n = \Delta$. So, if the system is resilient to $s$ stragglers then
\begin{align*}
s \leq \floor[\bigg] {n - \frac{n - n \gamma_u + 1}{n \gamma - n \gamma_u + 1}} = \floor[\bigg]{\frac{n^2 \gamma_c + n \gamma_u - 1}{n \gamma_c + 1}}.
\end{align*} 
\end{proof}
\begin{remark}
We emphasize that the construction in Algorithm \ref{Alg:Cyclic_Partially_Uncoded} is not isomorphic to a MDS code, even if $r_u = 1$.
\end{remark}
\begin{example}
\label{eg:coded_bottom_example}
Consider again the setting of Example \ref{eg:uncoded_cons} where $\Delta = n = 5$, $\gamma = \frac{3}{5}$. Suppose that we set $\gamma_u = \frac{2}{5}$ and $\gamma_c = \frac{1}{5}$.
This scheme is resilient to $\floor[\bigg]{\frac{n^2 \gamma_c + n \gamma_u - 1}{n \gamma_c + 1}} = 3$ stragglers and it can be verified that $\bfA \bfx$ can be computed once any $ Q = \Delta r_u - \frac{r_u}{2}( \ell_u + 1 ) + 1 = 8$ block rows have been processed. Thus, we can conclude that introducing a single coded block in each worker (at the bottom), helps to improve both $Q$ and the straggler resilience of the system as compared to an uncoded system.


\end{example}
\subsection{Coded blocks at the top}
The situation is quite different when we consider the placement of the coded blocks at the top of the worker nodes. In this case a given worker node only processes uncoded blocks after having processed $\ell_c$ coded blocks. Thus, if $x$ coded blocks have been processed by the worker nodes before starting work on the uncoded blocks, it suffices if {\it any} $\Delta - x$ distinct uncoded blocks are processed. This weaker requirement allows us to potentially improve the $Q_{ct}/\Delta$ ratio as compared to the previous constructions.

We now develop a lower bound on $Q_{ct}$. Suppose that we consider an arbitrary set of $\beta$ worker nodes that process all their blocks and another set of worker nodes that only contribute $x$ coded blocks. Evidently, in this case the total number of coded blocks is $x + \ell_c \beta$. Let $\calA$ denote the set of distinct uncoded blocks obtained from the chosen set of $\beta$ worker nodes. We note that
\begin{align*}
&Q_{ct} \geq x + \ell \beta + 1, \text{~when}\\
&x + \ell_c \beta + |\calA| < \Delta.
\end{align*}
This is because, we do not have enough equations to decode the $\Delta - |\calA|$ unknowns.

Next, we use another averaging argument. We calculate the average size of $\calA$ when considering all possible $\binom{n}{\beta}$ worker nodes via a double counting argument.


Consider a bipartite graph $G$, whose vertex set is $\calU \cup \calV$. Each element of $\calU$ is the set of uncoded blocks contained in a particular set of $\beta$ workers. Thus, the cardinality of $\calU$ is $|\calU| = {n \choose \beta}$. The set $\calV$ is the set of all possible uncoded blocks so that $|\calV| = \Delta$. There is an edge between $u \in \calU$ and $v \in \calV$ if $v \in u$. The degree of $v \in \calV$ in $G$ can be computed by observing that there are ${n - r_u \choose \beta}$ sets that do not contain $v$. Therefore, the degree of $v$ is ${n \choose \beta} - {n - r_u \choose \beta}$. Thus, the average degree of the nodes in $\calU$ is given by
\begin{align*}
\bar{d}_{\calU} = \Delta \times \left[ 1 - \frac{{n - r_u \choose \beta}}{{n \choose \beta}} \right].
\end{align*}
It follows that if
\begin{align*}
x + \ell_c \beta + \bar{d}_{\calU} < \Delta,
\end{align*}
 there is at least one choice of $\beta$ worker nodes that will not allow the decoding of $\bfA \bfx$. Our lower bound on $Q_{ct}$ can be derived by solving the following optimization problem.

\begin{equation}
\label{opt}
\begin{aligned}
\textrm{maximize} \; \; \; &  x + \ell \beta + 1 \\
\textrm{subject to} \; \; & \left( x +  \ell_c \beta \right) < \Delta  \left[\frac{{n - r_u \choose \beta}}{{n \choose \beta}} \right].
\end{aligned}
\end{equation}


\begin{example}
\label{eg:coded_top_example}


We can consider a $\langle n, \ell_u, \ell_c, \Delta, r_u \rangle = \langle 15, 3, 1, 15, 3\rangle$-top system and derive the bound $Q_{ct} \geq 18 > \Delta = 15$ by solving the optimization problem in (\ref{opt}). The optimal setting turns out to be $x = 1$ and $\beta = 4$.
\end{example}

\begin{example}
\label{eg:coded_top_example2}
Continuing our discussion of the setting in Examples \ref{eg:uncoded_cons} and \ref{eg:coded_bottom_example}, if we assign the coded blocks at the top rows of the workers and apply cyclic scheme for the uncoded portion, we can show that $Q_{ct} = 6$ while being resilient to $s=3$ stragglers. Thus, moving the coded rows to the top provides the best construction in terms of $Q$ and straggler resilience.
\end{example}

\section{Conclusion}
In this paper we have formulated a new model for distributed coded matrix-vector multiplication. Our model allows us to leverage partial work performed by stragglers while controlling the level to which coding is utilized in the solution.  We propose lower bounds on the required computation from the worker nodes in the worst case and present matching constructions in certain cases. Our proposed model demonstrates that the ordering of the computations within different worker nodes plays an important role in the overall job execution time. This in turn leads to new (to our best knowledge) code design problems that should be interesting to investigate.

%
%


\end{document}